\documentclass[10pt,twocolumn,twoside]{IEEEtran}
\IEEEoverridecommandlockouts

\usepackage{cite}
\usepackage{amsmath,amssymb,amsfonts}
\usepackage{graphicx}
\usepackage{color, soul}
\usepackage{textcomp}
\usepackage{balance}
\usepackage{xcolor}
\usepackage[utf8]{inputenc}
\usepackage{subcaption}
\usepackage{units}
\usepackage{enumitem}

\usepackage{algorithmic}
\usepackage{setspace}
\usepackage{url}
\usepackage[font=small]{caption}

\def\BibTeX{{\rm B\kern-.05em{\sc i\kern-.025em b}\kern-.08em
    T\kern-.1667em\lower.7ex\hbox{E}\kern-.125emX}}
\usepackage{lipsum}

\usepackage{framed}
 {\endMakeFramed}
\definecolor{shadecolor}{gray}{0.8}

\usepackage{empheq}
\definecolor{shadecolor}{gray}{0.8}

\title{\LARGE \bf 
Integrated Investment and Policy Planning for Power Systems via
Differentiable Scenario Generation
}

\author{Robert Mieth$^\dagger$
\vspace{-0.5cm}
 \thanks{ 
$^\dagger$RM is with the Department of Industrial and Systems Engineering at Rutgers University, New Brunswick, USA. {\tt robert.mieth@rutgers.edu}}
}

\usepackage{amsmath,amsfonts,amsthm,amssymb}
\usepackage{mathtools}
\usepackage[mathscr]{euscript}
\usepackage[bb=boondox]{mathalfa}
\usepackage[algoruled]{algorithm2e}
\usepackage{array}
\usepackage{mdframed}
\usepackage{bm}

\usepackage[noabbrev,capitalize]{cleveref}

\crefname{equation}{}{}
\Crefname{equation}{}{}
\crefname{section}{}{}
\Crefname{section}{}{}

\theoremstyle{definition} 

\theoremstyle{plain} 
\newtheorem{proposition}{Proposition}
 
\newtheorem{definition}{Definition}

\theoremstyle{remark} 

\usepackage{tikz}

\DeclareMathOperator{\proj}{proj}

\makeatletter
\newcommand{\pushright}[1]{\ifmeasuring@#1\else\omit\hfill$\displaystyle#1$\fi\ignorespaces}
\newcommand{\pushleft}[1]{\ifmeasuring@#1\else\omit$\displaystyle#1$\hfill\fi\ignorespaces}
\makeatother

\newcolumntype{C}[1]{>{\centering\arraybackslash}p{#1}} 

\newcounter{box}

\begin{document}

\begingroup
\allowdisplaybreaks

\maketitle

\begin{abstract}
We formulate a method to co-optimize power system capacity planning decisions and policy investments that shape electricity load patterns.
To this end, we leverage a gradient-based solution technique that enables the efficient solution of operation-aware planning models.
To compute gradients with respect to the conditions that define daily electricity demand profiles, we introduce and formalize the concept of differentiable scenario generation and show that generative machine learning models satisfy the mathematical requirements needed to compute consistent gradients.
We demonstrate the feasibility of the proposed approach through numerical experiments using a diffusion model–based scenario generator and a stylized generation and capacity expansion planning model.
\end{abstract}


\section{Introduction}

Electric load patterns are changing due to ongoing electrification efforts, the deployment of behind-the-meter energy resources, and the adoption of energy management systems~\cite{crozier2025distribution}.
As a result, investments in electric power infrastructure are necessary to ensure that power supply remains secure and affordable.
Traditional power system planning approaches consider load patterns as external scenarios that are simulated using predefined socio-technical assumptions, for example, the levels of building electrification or  electric vehicle adoption \cite{koltsaklis2018state}. 
These properties, however, are not arbitrary but can be influenced by policies like subsidies or tax credits \cite{penttinen2022regulatory}.
To date, there exists no method to co-optimize investments in such policies, their costs and their effects on electric load patterns, together with physical investments in generation and transmission assets, increasing the risk of  stranded assets.
This paper closes this gap via load scenarios created from conditional and differentiable scenario generation models that enable integrated policy and investment planning.

It is clear that grid investments are necessary to accommodate future electric loads \cite{cole2021quantifying,jenkins2021mission}.
Energy planning studies such as \cite{cole2021quantifying,jenkins2021mission} use pre-defined policy and socio-technical assumptions to create scenario sets that determine various investment pathways in generation and transmission. 
Analyses on investments in policy instruments that may shape these scenario sets exist largely in parallel to the energy planning and modeling community. However, it is also clear that energy technology adoption, and hence load patterns, can be influenced via policy and public investments \cite{penttinen2022regulatory,nadel2019electrification,li2019evolutionary}.

The compound effect of policy, technology adoption, consumer behavior, and environmental factors (e.g., temperature) on load patterns is a complex and stochastic process. However, modern computational tools allow simulating multiple thousand devices, households, or buildings and aggregate their (net-)load patterns at the desired spatial and temporal resolution \cite{mai2018electrification,crozier2025distribution,Priyadarshan2024EDGEi,earle2023impact,nazemi2025active}.
Integrating such detailed simulations into energy system planning models \textit{and} optimizing over the simulation parameters, however, is hopeless because the resulting model will be computationally intractable. 

Complementary to simulating load patterns, \textit{generative models} have emerged as a means to create time series data that can serve as an input to power and energy models. For example \cite{chen2018model,liang2022operation} propose conditional generative adversarial models. More recently, the work in \cite{lin2025energydiff} proposed the use of diffusion models. Options for energy system data synthesis via generative models were also discussed in a recent EPRI report~\cite{epri_synthetic_customer_load}.
Such generative models are typically based on a neural network architecture and are trained via gradient descent algorithms \cite{epri_synthetic_customer_load}.
This means, the generative load models are \textit{differentiable} both in terms of their parameters, but also in terms of their \textit{conditions}, i.e., the input that is given to such models to generate specific scenarios, for example season, day of the week, or temperature. 
This differentiability opens a pathway to iteratively tune the conditions of the scenario generator with the parameters of a planning model that is also solved via a gradient descent method, e.g., as proposed in \cite{degleris5169721gradient,ghazanfariharandi2026uncertaintyaware}.

This paper explores this possibility of co-optimizing conditions that shape load patterns and investments in power infrastructure. We make the following contributions:
\begin{enumerate}[leftmargin=*]
    \item We modify the operation-aware grid planning problem from \cite{degleris5169721gradient} such that it can co-optimize grid \textit{and} policy decisions that shape load scenarios via a gradient descent solution method.

    \item We formalize the concept of differentiable scenario generation and show that the resulting gradients are suitable to solve the target problem. 

    \item We describe and implement a proof-of-concept pipeline that simulates policy-dependent load scenarios, trains a suitable differentiable scenario generator, and applies these scenarios in a power system planning model.
\end{enumerate}

\section{Problem Formulation}

We consider a planner with the goal of future-proofing an electric power system. As common in power system planning, the planner seeks to optimize an objective $J(\eta)$ by making investment decisions $\eta$ within a set $\mathcal{I}$ of feasible investments. 
Objective $J$ and planning space $\mathcal{I}$ may encode emission targets, production cost targets, electrification targets, or limits on the potential for grid and generation expansion.

Given a planning decision $\eta$, the power system will pursue its day-to-day operations. Typically, given a set of available resources and an electricity demand $d$, the power system operator solves a cost-minimization problem of the form 
\begin{align}
    x^*(\eta, d) \in \arg\min_x \Big\{ c(x) : x\in\mathcal{X}(\eta, d)  \Big\}. 
\label{eq:operations_problem}
\end{align}
Decisions $x$ are the operational decisions, for example generator dispatches. Objective function $c$ captures the cost of running the system, for example generator fuel cost. 
The feasible space $\mathcal{X}$ is the feasible operational space of the system defined by generator capacities, complex generator constraints (e.g., ramping), power flow, transmission line capacities, and the requirement to balance generation and demand. 
We note that this \textit{operations problem} \eqref{eq:operations_problem} may be more complex in practice and entail multi-stage dispatch processes. The general formulation in $\mathcal{X}$ allows for these complexities by defining $\mathcal{X}$ as a set that depends on other optimization problems, but we assume single-stage operations moving forward.

Optimal planning subject to the operations becomes:
\begin{subequations}
\begin{align}
\min_{\eta\in\mathcal{I}} \quad 
    & J(\eta) \coloneqq \gamma_I^{\top} \eta + \mathbb{E}_{d\sim\mathbb{D}}\big[h\big(x^*(\eta, d)\big)\big] \label{eq:general_planning_objective}\\
\text{s.t.}\quad 
    & x^* \in \arg\min_x \Big\{ c(x) : x\in\mathcal{X}(\eta, d)  \Big\}, \label{eq:general_planning_operations}
\end{align}%
\label{eq:general_planing_model}%
\end{subequations}%
where function $h$ maps the operational decisions into the objective of the planner. 
This allows to capture the possibility that the planner is not strictly interested in cost-minimal dispatches but, for example, also values reliability, certain power flow patterns, or certain technology mixes in the dispatch. 

The expectation in \eqref{eq:general_planning_objective} is taken over the possible realizations of (net-)demand $d$ which we assume to follow a distribution $\mathbb{D}$.
Further, we model investment cost via linear cost coefficients $\gamma_I$.
We acknowledge that assuming linear investment costs may be restrictive, especially for discrete (``lumpy'') investment decisions. However, we highlight that additional complexity can be added here via piecewise linear functions or higher complexity $\mathcal{I}$.
For the purpose and exposition of this paper, linear investment costs are pertinent.

Traditional planning models focus on investments in physical assets, such as siting and sizing of new generation capacity, upgrading and hardening existing assets, and developing new transmission corridors. 
Commonly, these \textit{tangible} investments $\eta$ are optimized as a \textit{reaction} to the net-demand distribution $\mathbb{D}$. 
In other words, future net-demand patterns that inform the optimal choice of $\eta$ are considered external to the optimal planning model. 
Such patterns include emerging large loads such as data centers and, more critically, aggregated consumer behavior that results from changes in consumer technology. For example, electric vehicle charging, electric heat pumps, or electrified commercial buildings in combination with any energy management technology such as flexible smart charging or smart building energy management. 

It is clear that the economy-wide adoption of such load-shaping technologies is a distributed decision made by the individual consumers. However, this decision \textit{can} be shaped by policy decisions and \textit{intangible investments}, such as tax credits, subsidies, or specialized tariffs, that incentivize or discourage the adoption of certain technologies at the consumer level. 
We capture such policies in a vector $\pi$ and make the following subtle, but impactful, alteration to \eqref{eq:general_planing_model} that explicitly models the planner's ability to impact $\mathbb{D}$ via policy:
\allowdisplaybreaks
{
\setlength{\fboxsep}{0pt}
\begin{subequations}
\begin{empheq}[box=\colorbox{shadecolor}]{align}
\min_{\eta\in\mathcal{I}, \pi \in \mathcal{P}}  
    & J(\eta,\!\pi)\! \coloneqq\! \gamma_I^{\top} \eta \!+ \!\gamma_P^{\top} \pi\! + \!\mathbb{E}_{d\sim\mathbb{D}(\pi)}\big[h\big(x^*\!(\eta, d)\big)\big] \label{eq:da_planning_objective}\\
\text{s.t.}\  
    & 
    x^* \in \arg\min_x \Big\{ c(x) : x\in\mathcal{X}(\eta, d)  \Big\}
\end{empheq}%
\label{eq:demand_aware_planing_model}%
\end{subequations}}%
where, similar to cost related to $\eta$, we also assume a linear policy cost function with coefficients $\gamma_P$. 
The policy space $\mathcal{P}$ captures any constraints on the policy vector $\pi$.
Most notably, the net-demand distribution $\mathbb{D}(\pi)$ now depends on $\pi$, highlighting the impact of policy on demand patterns. 
In the following section, we will outline how problem \eqref{eq:demand_aware_planing_model} can be solved using a first-order gradient descent method. 

We note the the problem structure of \cref{eq:demand_aware_planing_model} is related to problems with decision-dependent uncertainty and in particular resembles \textit{performative prediction} \cite{perdomo2020performative,hardt2025performative,piliouras2023multi,mendler2020stochastic}, in particular the idea of using a functional relationship between model predictions and the data distribution \cite{mendler2022anticipating}.
In contrast to \cite{perdomo2020performative,hardt2025performative,piliouras2023multi,mendler2020stochastic,mendler2022anticipating}, however, we are not interested in predictive accuracy but assume to have knowledge about $\mathbb{D}(\pi)$ as we discuss below.

\section{Solution approach}

The problem in \eqref{eq:demand_aware_planing_model} cannot be solved directly. We need to deal with the complication of the expected value in objective \cref{eq:da_planning_objective} and the inner operations problem.

\subsection{Sample average formulation}

We first deal with the expectation operator as follows. 
Assume that we have access to a scenario generator $D_{\theta}(\pi, z, \epsilon)$ that takes as input the policy vector $\pi$, some additional \textit{context} $z$, for example, a temperature profile, and some random noise $\epsilon$.
Vector $\theta$ collects the parameters of the scenario generator.
In the following we may drop writing some of the explicit dependencies of $D$ on $\theta$, $z$, and $\epsilon$ and mostly write $D(\pi)$.  
Let $\hat{\mathbb{D}}^{N}_{\pi}$ be the empirical distribution supported by $N$ samples generated from $D(\pi)$.
For a sufficiently high number of samples $N$ the scenario generator should achieve $\hat{\mathbb{D}}^{N}_{\pi} \approx \mathbb{D}(\pi)$.
Given such a generator we denote $d_{\pi}^{(i)}$ as a scenario drawn from $D(\pi)$ and write the sample average version of \eqref{eq:demand_aware_planing_model} as:
\begin{subequations}
\begin{align}
\min_{\eta\in\mathcal{I}, \pi \in \mathcal{P}} \ 
    & \widehat{J}(\eta, \pi)\! \coloneqq\! \gamma_I^{\top} \eta\! +\! \gamma_P^{\top} \pi\! +\! \frac{1}{N}\sum_{i=1}^N h\big(x^*(\eta, d_{\pi}^{(i)})\big) \label{eq:saa_planning_objective}\\
\text{s.t.}\ 
    &  x^*(\eta, d_{\pi}^{(i)})\! \in \!\arg\min_x \Big\{\! c(x): x\!\in\!\mathcal{X}\big(\eta, d_{\pi}^{(i)}\big) \! \Big\}.
\end{align}%
\label{eq:saa_planing_model}%
\end{subequations}%
\vspace{-2em}

\subsection{Gradient Computations}
\label{ssec:gradient_computations}

Next, we tackle the bilevel structure that remains in \eqref{eq:saa_planing_model} by leveraging established results from implicit differentiation over the inner operations problem to formulate a stochastic gradient descent algorithm.
Specifically, we can write the gradients of the integrated investment and policy objective of \eqref{eq:saa_planing_model} as:
\begin{align}
    \nabla_{\eta} J(\eta, \pi) &= \gamma_I + \frac{1}{N}\sum_{i=1}^N \big(\nabla_\eta x^*_i \cdot \nabla_x h \big) \label{eq:grad_eta}\\
    \nabla_{\pi} J(\eta, \pi) &= \gamma_P + \frac{1}{N}\sum_{i=1}^N \big( \nabla_\pi d^{(i)} \cdot \nabla_d x^*_i \cdot \nabla_x h \big),
    \label{eq:grad_pi}
\end{align}
where we write $x_i^*$ as shorthand for $x^*(\eta, d_{\pi}^{(i)})$.
The gradient $\nabla_x h$ is straightforward to compute assuming that $h$ is differentiable. 
The Jacobian $\nabla_\eta x^*_i$ can be computed using modern approaches from implicit function differentiation.
Implicit differentiation has been discussed comprehensively, e.g., in \cite{amos2017optnet,besanccon2024flexible}. Energy applications are discussed, e.g., in \cite{degleris5169721gradient,donti2017task}. 
We refer the reader to these references for details on the theory and only repeat the intuition of how to obtain $\nabla_\eta x^*_i$ here. 
In fact, a formulation for $\nabla_\eta x^*_i$ from an optimal primal-dual solution of the operations problem can be derived from the problem's KKT system using classic results from sensitivity analysis~\cite{castillo2007closed}. 
The strength of modern differentiable optimization comes from the realization that instead of first computing $\nabla_\eta x^*_i$ and \textit{then} multiplying it with $\nabla_x h$, we can directly solve the KKT system for $(\nabla_\eta x^*_i \cdot \nabla_x h)$. 
This way, it is possible to exploit the structure of the KKT system to simplify obtaining necessary derivatives by backpropagating through the problem solution via automatic differentiation \cite{amos2017optnet,besanccon2024flexible}.

We note that the effectiveness of implicit differentiation depends on the properties of the inner operations problem. 
Usually, strongly convex problems like quadratic or conic problems work best as they provide unique optimal solutions. However, our experiments with linear $c$ and $\mathcal{X}$ were successful.

For \eqref{eq:grad_pi}, we need the Jacobian $\nabla_d x_i^*$.
Generally, this Jacobian will be easier to compute than $\nabla_\eta x^*_i$ because the net-demand usually appears on the right-hand side of the constraints. This allows re-using the KKT factorization from the optimal solution \cite{pacaud2025sensitivity}. As a result, right-hand side derivatives can be computed virtually for free from an optimal primal-dual solution.
If, however, $d$ also appears as a parameter in the constraint matrix, the same argument holds as for $\nabla_\eta x^*_i$.

Finally, computing a gradient $\nabla_\pi d$ implies a functional relationship between $\pi$ and $d$. So far, we have only established this relationship for the simulator $D(\pi)$. 
To close the connection between the demand simulator and the differentiability of a load scenario $d_{\pi}^{(i)}$ we introduce the following definition:

\begin{definition}[Differentiable Scenario Generator]\label{def:diff_scen_gen}
We call a function $D_{\theta}(\pi, z, \epsilon)$ a \emph{differentiable scenario generator} if it has the following properties. 
\begin{enumerate}[leftmargin=*]
    \item Let $\pi^{\star}$ and $z^{\star}$ be given fixed parameters and let $\epsilon$ be a  random variable $\epsilon$ with fixed distribution $\mathcal{Q}$, e.g., $\mathcal{Q}=\mathcal{N}(0,I)$.
    Further, for a set of samples $\{\epsilon^{(i)}\}_{i=1}^N$ from $\epsilon$, write $d_{\pi^{\star}}^{(i)} = D_{\theta}(\pi^{\star}, z^{\star}, \epsilon^{(i)})$ and let $\{d_{\pi^{\star}}^{(i)}\}$ be the set of corresponding scenarios. For sufficiently large $N$ the empirical distribution $\hat{\mathbb{D}}_{\pi^\star}^N \coloneqq \frac{1}{N}\sum_{i=1}^N \delta_{d_{\pi^{\star}}^{(i)}}$ achieves $\hat{\mathbb{D}}_{\pi^\star}^N \approx \mathbb{D}(\pi^{\star})$, i.e., similarity to a target distribution. 

    \item For a given fixed sample  $\epsilon^{\star}$ of $\epsilon$ the gradient $\nabla_{\pi^{\star}} D_{\theta}(\pi^{\star}, z^{\star}, \epsilon^{\star})$ exists and is computable.
\end{enumerate}
We define  $\nabla_{\pi} d_{\pi}^{(i)} \coloneqq \nabla_{\pi} D_{\theta}(\pi, z, \epsilon^{(i)})$ for a given sample $\epsilon^{(i)}$.
\end{definition}

We note that Definition~\ref{def:diff_scen_gen} implies that, besides the random input $\epsilon$, the scenario generator is a deterministic function, i.e., for fixed $\pi^{\star}$ and $z^{\star}$ we have that $D(\pi^{\star}, z^{\star}, \epsilon) = D(\pi^{\star}, z^{\star}, \epsilon')$ for any $\epsilon = \epsilon'$.
These properties are common in modern generative models as we show and discuss in Section~\ref{sec:numerical_experiments} below.

\subsection{Stochastic Gradients}

We first argue that taking the scenario sample gradients is indeed a valid estimator for the change of expected cost with regard to the distributional shift resulting from a change in $\pi$.
\begin{proposition}
    Write $f(d) = h(x(\eta, d))$ as a shorthand.
    For a given differentiable scenario generator $D(\pi, \epsilon)$ (as per Definition~\ref{def:diff_scen_gen}) the gradient  $\frac{1}{N}\sum_{i=1}^N \nabla_{\pi}f\big(D(\pi, \epsilon^{(i)})\big)$ is an unbiased estimator of $\nabla_{\pi} \mathbb{E}_{d\sim \mathbb{D}(\pi)}\big[f(d)\big]$. 
\end{proposition}
\begin{proof}
The proof follows the argument from \cite[Sec.~5.2]{mohamed2020monte} and requires $D(\pi, \epsilon)$ to be differentiable or, to use standard machine learning language, have a \emph{differentiable sampling path}. This assumption holds as per Definition~\ref{def:diff_scen_gen}.
We introduce $\omega(x, \pi)$ as the density function of $\mathbb{D}(\pi)$ and for this proof use the symbol $x$ for scenarios sampled from $D(\pi, \epsilon)$ to avoid confusion with the differential operator and to remain consistent with most literature on generative models.
Further, we write $q(\epsilon)$ as the density function of the distribution $\mathcal{Q}$ of the random variable $\epsilon$.
We can now write
\begin{align*}
\nabla_{\pi} \mathbb{E}_{x\sim \mathbb{D}(\pi)}\big[f(x)\big]   
    &= \nabla_{\pi} \int \omega(x, \pi) f(x) dx \\
    &= \nabla_{\pi} \int q(\epsilon)f\big(D(\pi, \epsilon)\big)d\epsilon \\
    &= \mathbb{E}_{\epsilon\sim\mathcal{Q}}\Big[ \nabla_{\pi}f\big(D(\pi, \epsilon)\big)) \Big].  
\end{align*}
It is clear that the final expectation can be estimated through $\frac{1}{N}\sum_{i=1}^N \nabla_{\pi}f\big(D(\pi, \epsilon^{(i)})\big)$ which is exactly the structure required for the gradient in \eqref{eq:grad_pi}.
\end{proof}

\subsection{Gradient descent algorithm}
\label{ssec:gradient_descent_algo}

With the gradients available, we can solve \eqref{eq:saa_planing_model} using a stochastic projected gradient descent approach \cite{degleris5169721gradient}. 
For an iteration counter $\tau$ and a batch of $N_{\rm B}$ samples we can  write the stochastic estimate of the gradients from \cref{eq:grad_eta,eq:grad_pi} as
\begin{align}
    \Delta_{\eta}^{(\tau)} &= \gamma_I + \frac{1}{N_{\rm B}}\sum_{i=1}^{N_{\rm B}} \big(\nabla_\eta x^*_i \cdot \nabla_x h \big) \label{eq:stoch_grad_eta}\\
    \Delta_{\pi}^{(\tau)}  &= \gamma_P + \frac{1}{N_{\rm B}}\sum_{i=1}^{N_{\rm B}} \big( \nabla_\pi d^{(i)} \cdot \nabla_d x^*_i \cdot \nabla_x h \big).
    \label{eq:stoch_grad_pi}
\end{align}
For a given learning rate $\lambda$, we can perform the gradient steps
\begin{align}
    \eta^{(\tau+1)} &= \proj_{\mathcal{I}}\big(\eta^{(\tau)} - \lambda \Delta_{\eta}^{(\tau)}\big) \\
    \pi^{(\tau+1)} &= \proj_{\mathcal{P}}\big(\pi^{(\tau)} - \lambda \Delta_{\pi}^{(\tau)}\big),
\end{align}
until a convergence criterion has been reached.
The operator $\proj_{\mathcal{X}}$ denotes the projection onto set $\mathcal{X}$. The projection ensures feasibility of the investments in each step.

\section{Numerical Experiments}
\label{sec:numerical_experiments}

We conduct numerical experiments to establish a proof-of-concept for using differentiable scenario generation. 
For this purpose we study the dependence of daily load patterns as a function of the policy vector $\pi = (\pi^{\rm EV}_{\rm adopt}, \pi^{\rm EV}_{\rm flex}, \pi^{\rm HP}_{\rm adopt}, \pi^{\rm HP}_{\rm eff})$:
\begin{itemize}[leftmargin=*]
    \item $\pi^{\rm EV}_{\rm adopt}\in[0,1]$ describes the level of EV adoption and resulting charging demand.

    \item $\pi^{\rm EV}_{\rm flex}\in[0,1]$ describes the level of adoption of EV smart charging and flexibility technologies and captures if EV charging demand is shifted away from peak demand hours.

    \item $\pi^{\rm HP}_{\rm adopt}\in[0,1]$ describes the level of heat pump (HP) adoption and building electrification. It captures the sensitivity of the load to colder temperatures alongside a general shift of base and peak demand.

    \item $\pi^{\rm HP}_{\rm eff}\in[0,1]$ describes the level of building efficiency and reduces direct temperature sensitivity.
\end{itemize}


\subsection{Simulation setup}
\label{ssec:simulation_setup}

The objective of the simulation is to create a set of sample tuples $\mathcal{S}_M = \{(d^{(m)}, \pi^{(m}, z^{(m)})\}_{m=1}^M$ that we can use to train the parameters $\theta$ of scenario generator $D_{\theta}(\pi, z, \epsilon)$. This simulator can be arbitrarily complex and does not have to meet any requirements on differentiability. 
Detailed models that simulate individual loads and their behavioral, contextual (e.g., temperature-dependent), and otherwise stochastic patterns are  available and well-studied \cite{mai2018electrification,crozier2025distribution,Priyadarshan2024EDGEi,earle2023impact,nazemi2025active}.

For the purpose of this paper we created a simple scenario generator that qualitatively captures the EV and heat-pump driven load patterns reported in \cite{crozier2025distribution}.
As a data basis, we use three years (from 2023, 2024, 2025) of daily 24h load patterns obtained from PJM \cite{pjmdataminer} alongside hourly temperature data at the center of the respective PJM zone. We obtained temperature data from the Copernicus ERA5 reanalysis data~\cite{C3S_ERA_2025}.
We call the real-world load data \textit{base load}. 

We created a total of $M=10,000$ training samples. 
For each $m\in\{1,...,M\}$ we selected a random base load alongside its corresponding temperature time series.
Next, we sampled a random $\pi^{(m)}$ by drawing uniformly from the unit hypercube~$[0,1]^4$.
For the random policy $\pi^{(m)}$ and the temperature profile, which we capture in the context vector $z^{(m)}$, we compute the policy-dependent load profile as follows.

\subsubsection{EV adoption ($\pi^{\rm EV}_{\rm adopt}$)} Increasing EV adoption will add peak load to the base load of up to 75\% of the current yearly peak load \cite{crozier2025distribution}. We simulate this by adding bell curve-shaped load on top of the base load profile centered at the 18-th hour of the day with a standard deviation of 1h. The height of the bell curve is scaled continuously between 0 ($\pi^{\rm EV}_{\rm adopt}=0$) and 0.75 ($\pi^{\rm EV}_{\rm adopt}=1$), given relative to the yearly peak load. For each sample the exact center of the peak and the exact standard deviation of the bell curve is further distorted by multiplication with random noise with standard deviation $0.05$ and $0.1$, respectively.
In addition, $\pi^{\rm EV}_{\rm adopt}$ creates additional uniform peak and off-peak load drawn from $\mathcal{N}(0.1,0.01)$.

\subsubsection{EV flexibility ($\pi^{\rm EV}_{\rm flex}$)} Increased EV flexibility widens the bell curve that models EV load and moves the peak load occurrence away from the peak at the 18-th hour. 
At $\pi^{\rm EV}_{\rm flex}=1$ the standard deviation of the added load is widened to 4h and the peak is moved at most 5h from the 18-th hour.
The same random noise factor as described above applies to the flexibility-scaled load.
We note that any load addition that would occur past midnight is applied to the morning hours of the same day to ensure consistency.

\subsubsection{HP Adoption ($\pi^{\rm HP}_{\rm adopt}$)}
HP adoption adds load as a function of the temperature profile up to an additional peak load of 2 relative to the yearly peak load \cite{crozier2025distribution}.
We compute the electricity demand from heating for each hour of the day by taking the weighted average of the hourly temperature (80\% weight) and the daily mean temperature (20\% weight). The heat demand signal is then computed to rise smoothly whenever the temperature signal $T_t$ for each hour $t$ falls below $T_{\rm th} = 19^\circ\rm{C}$ via $\tau \log\big(1+\exp\big(\frac{T_{\rm th} - T_t}{\tau}\big)\big)$, with $\tau=1.3$.
The resulting heating demand is then scaled with the level of HP adoption $\pi^{\rm HP}_{\rm adopt}$.
We introduce random behavior by drawing a random smoothing factor from $\mathcal{N}(0.05, 0.07)$. 
The smoothing factor computes the convex combination of the hourly heat demand signal and its daily mean. If the smoothing factor is equal to 1, then the signal is equal to its mean. 

\subsubsection{Building efficiency ($\pi^{\rm HP}_{\rm eff}$)}
Increasing building efficiency reduces peak demand from heating and shifts the demand dependency from hourly temperature to daily mean temperature. 
We model this by continuously increasing the mean and standard deviation of the heat demand smoothing factor up to 0.8 and 0.3, respectively, for $\pi^{\rm HP}_{\rm eff}=1$.

Fig.~\ref{fig:scenarios} shows some exemplary simulations for various settings of $\pi$.
We highlight that the central purpose of this simulation setup is \textit{not} to realistically model policy-dependent demand behavior. Better tools are available and should be used for that. 
The purpose of our simulation setup was to create a low-complexity but still meaningful scenario generator to test the feasibility of differentiable scenario generation. 

\subsection{Differentiable scenario generator}
\label{ssec:differentiable_scenario_generator}

\begin{figure}
    \centering
    \includegraphics[width=0.95\linewidth]{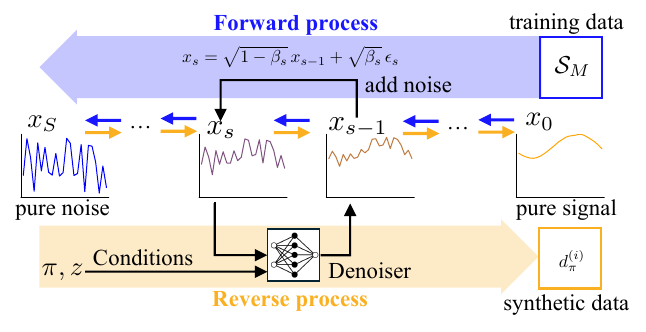}
    \caption{Overview of the forward and reverse diffusion processes. The forward process adds noise over $S$ diffusion steps on the available training data. A denoising model is trained to learn how to remove noise from a given noisy signal using given additional conditions and the current step index $s$. 
    }
    \label{fig:diffusion_overview}
\end{figure}

We use a diffusion model (DM) as an effective way to train a differentiable scenario generator. 
We outline the relevant mechanics of the DM below and also refer to \cite{lin2025energydiff,yang2024survey} for more details on using DMs for time series synthesis.

The fundamental idea of a DM is to train a \textit{denoiser model} to effectively remove noise from a noisy signal until all that is left is signal. See Fig.~\ref{fig:diffusion_overview} for an overview.
The noise removal is step-wise over a number of $s=0,...,S$ steps.
To this end, a \textit{forward process} corrupts a given signal $x_s$ via 
    $x_s = \sqrt{1-\beta_s}x_{s-1} + \sqrt{\beta_s} \epsilon_s$,
where $x_0$ is the pure starting signal, $x_S$ is the final corrupted signal close to pure noise, and $\epsilon_s$ is random gaussian noise.
We note that for the discussion in this section we also use symbol $x$ for the time series vector to remain in line with standard DM language.
Parameters $\beta_s$ define the diffusion \textit{schedule} and are chosen such that noise is not added too quickly in training.

In the reverse process, the denoiser model  takes newly sampled noise $\epsilon$ as an input, alongside conditions $\pi$, $z$, and information on the step number $s$ in the denoising process.
After $S$ applications of the denoiser model on the noisy signal, the final result is a synthetic data sample. 
The reverse process is the same for training and for sampling new scenarios.
Algorithms~\ref{alg:training} and \ref{alg:sampling} in Appendix~\ref{apx:algos} detail the training and sampling procedures, respectively. 
We highlight that our target scenario generator $D_{\theta}(\pi, z, \epsilon)$ is the \textit{entire} sampling process as shown in Algorithm~\ref{alg:sampling}.
This involves applying the trained denoiser model $S$ times. 
In our case study, we use a multi-layer perceptron with a single hidden layer of width 128.

Because our goal is to effectively compute gradients over the sampling condition $\pi$ (i.e., $\nabla_{\pi}d_{\pi}$ as discussed in Section~\ref{ssec:gradient_computations} above), we slightly depart from traditional DM design where new noise is added in each reverse step $s$.
Specifically, we adopt the approach from \cite{song2020denoising} where the reverse process is a deterministic path from the initial noise sample ($x_S$) to the final target scenario ($x_0$).
This way, we can compute well-behaved gradients $\nabla_{\pi}d_{\pi}$ and, as a side-effect, sample faster.


Finally, we note that for our case study we trained the DM on load residuals. This means in addition to the temperature profile, context vector $z$ also contains a base load profile. The DM then learns how this base load profile is altered to reflect the policy-induced patterns. Because we are simulating the load patterns using real-world temperature and load data, this approach was pertinent and provided much better results in the scenario generation than directly learning the load patterns. 
We also highlight that all data was normalized before training and needs to be de-normalized after data generation. This is a standard approach in machine learning pipelines.


\subsection{Differentiable scenarios}

\newcommand{\fw}{0.9}

\begin{figure}
\begin{subfigure}[t]{\linewidth}
    \centering  
    \includegraphics[width=\fw\linewidth]{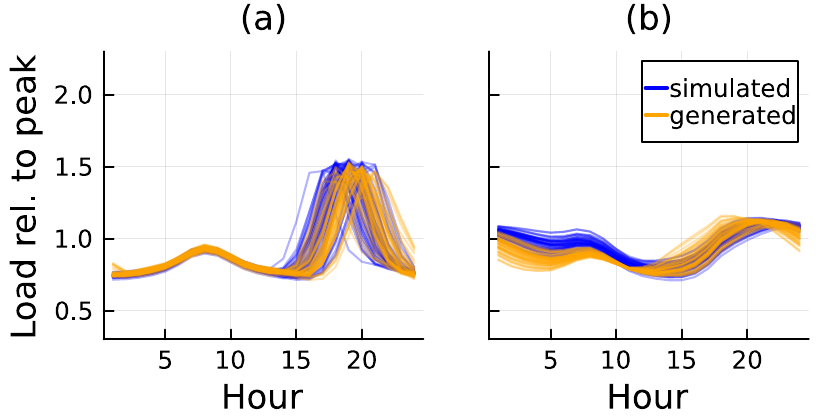}
\end{subfigure}\vspace{0.5em}
\begin{subfigure}[t]{\linewidth}
    \centering  
    \includegraphics[width=\fw\linewidth]{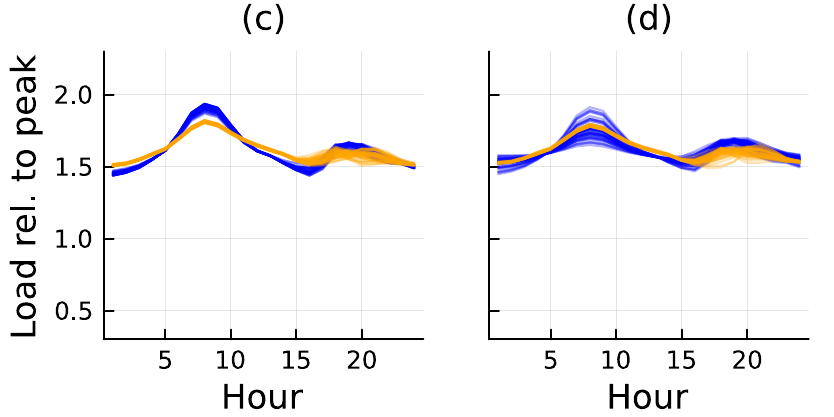}
\end{subfigure}\vspace{0.5em}
\begin{subfigure}[t]{\linewidth}
    \centering  
    \includegraphics[width=\fw\linewidth]{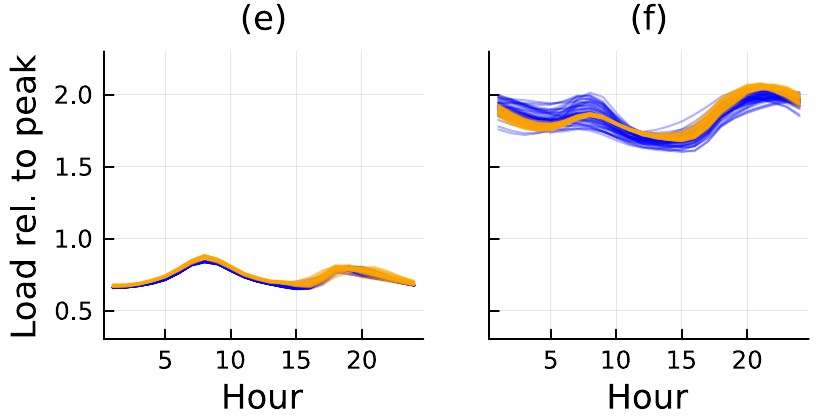}
\end{subfigure}
\caption{
    Simulated and generated scenarios for the same day (i.e., for a given load baseline and winter temperature profile as in Fig.~\ref{fig:grad_HP}) for different policy vectors ($\cdot=0$).
    (a) $\pi=[1,\cdot,\cdot,\cdot]$ (max. EV adoption). 
    (b) $\pi=[1,1,\cdot,\cdot]$ (max. EV adoption with max. flexibility). 
    (c) $\pi=[\cdot,\cdot,\cdot, 1]$ (max. HP adoption).
    (d) $\pi=[\cdot,\cdot, 1, 1]$ (max. HP adoption with max. building efficiency).
    (e) $\pi=[\cdot,\cdot,\cdot,\cdot]$ (baseline).
    (f) $\pi=[1,1,1,1]$ (max. electrification with max. flexibility and efficiency).
    } \label{fig:scenarios}
\end{figure}

\begin{figure}
    \centering
    \includegraphics[width=0.85\linewidth]{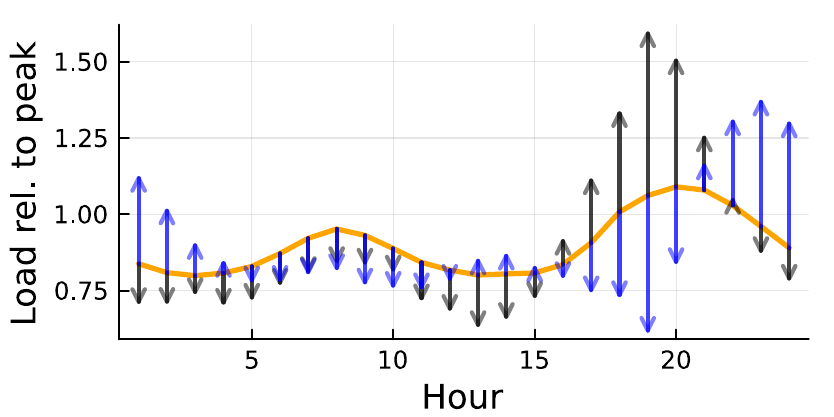}
    \caption{One generated scenario for $\pi=[0.5, 0.5, 0.1, 0.1]$ and the gradients corresponding to the components $\pi^{\rm EV}_{\rm adopt}$ (black arrows) and $\pi^{\rm EV}_{\rm flex}$ (blue arrows). The scenario gradient correctly reflects the simulated behavior that an increase of $\pi^{\rm EV}_{\rm adopt}$ will increase peak demand while an increase of $\pi^{\rm EV}_{\rm flex}$ will reduce peak demand and shift demand towards the night.}
    \label{fig:grad_EV}
\end{figure}

\begin{figure}
    \centering
    \includegraphics[width=\fw\linewidth]{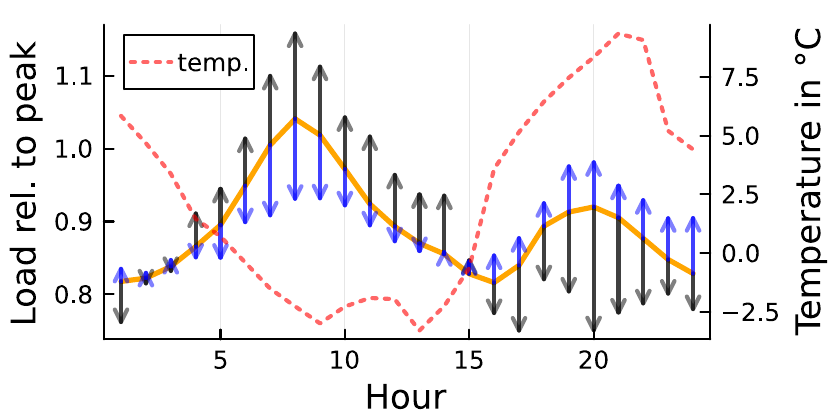}
    \caption{One generated scenario for $\pi=[0.1, 0.1, 0.5, 0.5]$ and the gradients corresponding to the components $\pi^{\rm HP}_{\rm adopt}$ (black arrows) and $\pi_{\rm HP}^{\rm eff}$ (blue arrows) alongside the temperature profile of that day (red dashed line).
    The scenario gradient correctly captures the simulated opposing impact of $\pi^{\rm HP}_{\rm adopt}$ and $\pi_{\rm HP}^{\rm eff}$ on the temperature dependency of the load.
    }
    \label{fig:grad_HP}
\end{figure}

Fig.~\ref{fig:scenarios} shows an overview of some exemplary scenarios, both \textit{simulated} with the simulation setup described in Section~\ref{ssec:simulation_setup} and \textit{generated} with the scenario DM-based scenario generator described in Section~\ref{ssec:differentiable_scenario_generator}.
All plots in Fig.~\ref{fig:scenarios}(a)--(f) have been created for the same day, i.e., the same baseline load and temperature profile. The temperature profile is plotted in Fig.~\ref{fig:grad_HP}. Each subplot shows the simulated and generated scenarios for a different target policy vector. 
We observe that the generated profiles closely match the simulated samples. 
We note that the scenarios in Fig.~\ref{fig:scenarios} are ``out-of-sample'' in the sense that both the simulation and the scenario generation was performed after training with new inputs and new noise.

Figs.~\ref{fig:grad_EV} and \ref{fig:grad_HP} show the gradient of the scenarios with respect to a subset of the components of the policy vector $\pi$. 
We observe that the gradients obtained from this single sampling instance correctly reflect how load is reshaped through $\pi$ in the simulations.
For example, in Fig.~\ref{fig:grad_EV} we see that an increase of EV adoption ($\pi^{\rm EV}_{\rm adopt}$) will increase peak load at the 18-th hour and increase ramping steepness in the hours leading up to the peak hour. 
Likewise, an increase of EV flexibility ($\pi^{\rm EV}_{\rm flex}$) will decrease the load at the peak hour and move load towards the off-peak night hours. 
The analysis of the gradients related to HP adoption and building efficiency ($\pi^{\rm HP}_{\rm adopt}$, $\pi^{\rm HP}_{\rm eff}$) in Fig.~\ref{fig:grad_HP} correctly captures the opposing effect of these two parameters on the temperature sensitivity of the load. 
We highlight again that this information is completely encoded in a single scenario $d_{\pi}^{(i)}$, i.e., in a single parametrization of $D_{\theta}(\pi, z, \epsilon)$.

\subsection{Planning application}
\label{ssec:planning_application}

\begin{figure}
\begin{subfigure}[t]{\linewidth}
    \centering  
    \includegraphics[width=\fw\linewidth]{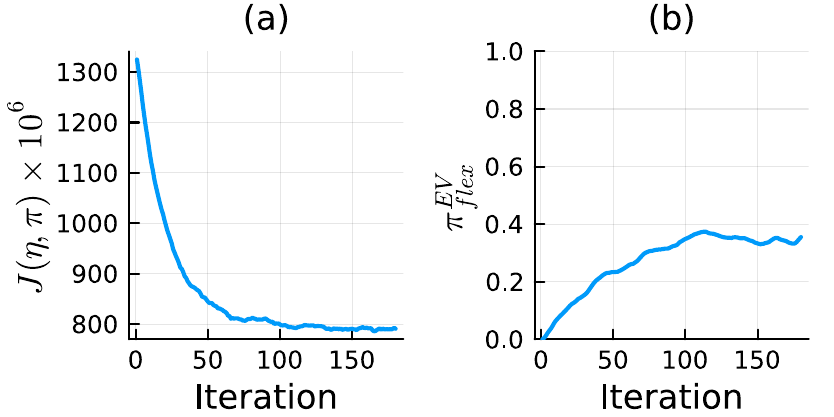}
\end{subfigure}\vspace{0.5em}
\begin{subfigure}[t]{\linewidth}
    \centering  
    \includegraphics[width=\fw\linewidth]{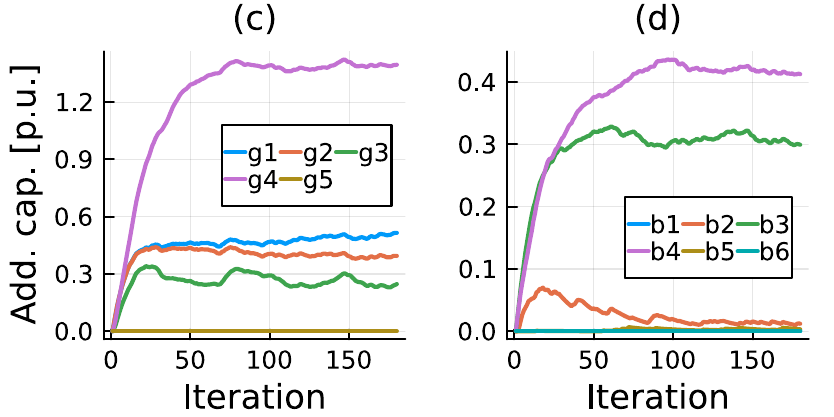}
\end{subfigure}\vspace{0.5em}
\caption{
    Results from solving the planning model in \eqref{eq:detailed_planning_model} via gradient descent (Sec.~\ref{ssec:gradient_descent_algo}).
    Plots show the trajectories of (a) the planning objective, (b) policy $\pi^{\rm EV}_{\rm flex}$, (c) generator capacity additions, and (d) transmission (branch) capacity additions.
    } \label{fig:planning_results}
\end{figure}

We now use our trained scenario generator in a stylized capacity expansion planning model.
The goal is to accommodate 100\% EV adoption ($\pi_{\rm adopt}^{\rm EV}=1$), which requires generation and transmission capacity  investments due to the resulting load increase. 
To this end, the decision maker balances investments in generation capacity and transmission line upgrades and investments in enabling EV flexibility (given by $\pi_{\rm flex}^{\rm EV}$). 
Model \eqref{eq:detailed_planning_model} in Appendix~\ref{apx:detailed_planning_model} shows the detailed formulation.

We use the PJM 5 bus model with the same parametrization as in \cite{mieth2024prescribed} based on the MATPOWER \texttt{case5} data set \cite{matpowercase5}. 
We use the single-period load vector from the dataset and multiply it with the generated load scenarios scaled to the interval~$[0,1.3]$. 
The operation model is run for 24-hour scenarios in batches of 5 days sampled randomly. 
If the available generation and transmission capacities prohibit a feasible load-serving dispatch, generation and transmission constraints are relaxed at a high penalty. By investing in additional capacity, the model can avoid this penalty. 
We set annualized cost of generation and transmission upgrades of \$1m/MW. Achieving $\pi^{\rm EV}_{\rm flex}=1$ costs \$1,100m. We set the learning rate to $\lambda=10^{-9}$ and initial investments and the initial $\pi^{\rm EV}_{\rm flex}$ to zero.

Fig.~\ref{fig:planning_results} shows the results. The model converges after 180 iterations in under one minute. Fig.~\ref{fig:planning_results}(a) shows the planning objective. 
The final optimal level of $\pi^{\rm EV}_{\rm flex}$ is $36\%$ as per Fig.~\ref{fig:planning_results}(b).
Figs.~\ref{fig:planning_results}(c) and (d) show the resulting additional capacity investments for each generator (g1--g5) and each transmission branch (b1--b6).

\subsection{Implementation and code availability}

We implemented all necessary code for our experiments in Julia. For the DM and its training algorithm we used the \texttt{Flux.jl} \cite{Flux.jl-2018} package. For implementing and solving the optimization problem we used the \texttt{JuMP} \cite{DunningHuchetteLubin2017} package and for differentiation of the optimization model we used the \texttt{DiffOpt.jl} \cite{besancon2023diffopt} package.
All computations were performed on an Apple MacBook Pro with an Apple M3 Pro processor and 36 GB of memory.

Our data and code are available open-source at:
\begin{center}
    [\textit{\url{github.com/ropes-lab/differentiable_scenario_panning}}].
\end{center}

\section{Conclusion and Discussion}

This paper introduced differentiable load scenarios to co-optimize load shape (e.g., determined by policy) and other planning decisions (e.g., capacity investments) and demonstrated the feasibility of the concept.
To this end, we formulated a suitable planning problem, formalized the concept of a differentiable scenario generation, and discussed necessary mathematical properties. 
Our numerical experiments showcased a lightweight simulation and conditional scenario generation pipeline using a diffusion model architecture. 
The scenarios \textit{and} their resulting gradients showed the expected behavior and recovered the complex underlying dependencies between scenario conditions and the load shape itself. 
Finally, we applied the concept to a stylized capacity expansion planning problem and showed the feasibility of the approach. 

We highlight again that this work is a proof-of-concept. More research is needed to develop a framework that allows application to quantitative planning models. 
First, a more detailed simulator should be used. Such simulators are available but still require a carefully crafted data pipeline to lend themselves to training data generation. 
Second, our proposed scenario generator is low complexity. Generative models with architectures that are more carefully tailored for time series have been discussed in the literature \cite{yang2024survey}. A systematic study on the tradeoff between scenario accuracy and differentiability for different model architectures is needed. 
Finally, more research on scalability is required. Our case study solved quickly and was implemented without meaningful optimization or parallel computation and gradient-based planning models have been shown to scale \cite{degleris5169721gradient}. However, more experiments and research are needed to identify optimally scalable formulations and implementations of our approach.



\appendix 

\subsection{Training and sampling algorithms}
\label{apx:algos}

Algorithms~\ref{alg:training} and \ref{alg:sampling} show details on the DM training and the sampling procedure.

\begin{algorithm}[h!]
\caption{Training process}\label{alg:training}
\begin{algorithmic}[1]
\REQUIRE Training dataset $\mathcal{S}_M$, learning rate $\lambda$, schedule~$\{\beta_s\}_{s=1}^S$, batch size $N_{\rm B}$
\STATE Define $\alpha_s = 1-\beta_s$ and $\bar{\alpha}_s=\prod_{j=1}^s \alpha_j$
\STATE Initialize denoiser parameters $\theta$
\WHILE{not converged}
    \STATE Sample minibatch $\{(x^{(i)}, \pi^{(i)}, z^{(i)})\}_{i=1}^B$ from $\mathcal{S}_M$
    \FOR{$i=1,\dots,N_{\rm B}$}
        \STATE Sample random diffusion step $s_i\!\sim\! \mathcal{U}\{1,...,S\}$
        \STATE Sample noise $\epsilon_i \sim \mathcal{N}(0,I)$
        \STATE Forward diffuse:
        $x_{s_i}^{(i)} = \sqrt{\bar{\alpha}_{s_i}}\,x_0^{(i)}\!+\! \sqrt{1-\bar{\alpha}_{s_i}}\,\varepsilon_i$
        \STATE Predict signal:
        $\hat{x}_{0,\theta}^{(i)} = \mu_\theta\!\left(x_{s_i}^{(i)}, \pi^{(i)}, z^{(i)}, s_i\right)$
    \ENDFOR
    \STATE Compute minibatch loss: $\mathcal{L}(\theta)=\frac{1}{B}\sum_{i=1}^B \left\|x_0^{(i)}-\hat{x}_{0,\theta}^{(i)}\right\|_2^2$
    \STATE Update parameters: $\theta \leftarrow \theta - \lambda \nabla_\theta \mathcal{L}(\theta)$
\ENDWHILE
\RETURN trained denoiser $\mu_{\theta}$
\end{algorithmic}
\end{algorithm}

\begin{algorithm}[h!]
\caption{Sampling process with fixed noise}\label{alg:sampling}
\begin{algorithmic}[1]
\REQUIRE Trained denoiser $\mu_\theta$, conditions $(\pi, z)$, noise schedule $\{\beta_s\}_{s=1}^S$
\STATE Define $\alpha_s = 1-\beta_s$ and $\bar{\alpha}_s=\prod_{j=1}^s \alpha_j$
\STATE Draw noise $\epsilon$ where $\epsilon \sim \mathcal{N}(0,I)$
\STATE Initialize $x_S \leftarrow \epsilon$
\FOR{$s = S, S-1, \dots, 1$}
    \STATE Reverse process:
    $\hat{x}_0 = \mu_\theta(x_s, \pi, z, s)$
    \STATE Compute reverse update:
    $x_{s-1} =
    \frac{1}{\sqrt{\alpha_s}}
    \left(
        x_s - \frac{1-\alpha_s}{\sqrt{1-\bar{\alpha}_s}}
        (x_s - \sqrt{\bar{\alpha}_s}\,\hat{x}_{0})
    \right)$
\ENDFOR
\RETURN $x_0$
\end{algorithmic}
\end{algorithm}

\subsection{Detailed planning model}
\label{apx:detailed_planning_model}

The planning model that we used for the experiments in Section~\ref{ssec:planning_application} follows the form of \eqref{eq:saa_planing_model}. In detail:
\begin{subequations}
\begin{align}
\min_{\eta_{\rm G}, \eta_{L}, \pi^{\rm EV}_{\rm flex}}\ & \gamma_{\rm I,G}^{\top}\eta_{\rm G} + \gamma_{\rm I,L}^{\top}\eta_{\rm L} + \gamma_{\rm P}(\pi^{\rm EV}_{\rm flex}) + 365 \frac{1}{N}\sum_{i=1}^Nc(x_i^*) \\
\text{s.t.}\ 
& 0\le \pi^{\rm EV}_{\rm flex} \le 1 \label{eq:pi_lims}\\
& \eta^{\rm G}, \eta^{L} \ge 0 \label{eq:inv_lims}\\
& x_i^* \in \Big\{ 
    \arg\min_{x}\ \sum_{t=1}^{24}\big( c^{\top}p_t + \rho^{\top}_{\rm G}s_{\rm G,t} + \rho^{\top}_{\rm F}s_{\rm F,t} \big) \label{eq:op_objective} \\
    & \qquad\quad  \mathbb{1}^{\top}p_t = \mathbb{1}^{\top}d^{(i)}_{\pi,t}\quad \forall t \label{eq:op_enerbal}\\
    & \qquad\quad p_t \le p^{\rm max} + \eta_{\rm G} + s_{G,t}\quad \forall t \label{eq:op_gen_lims}\\
    & \qquad\quad f_t = B(p_t - d^{(i)}_{\pi,t})\quad \forall t \label{eq:op_dc_pf}\\
    & \qquad\quad |f_t| \le f^{\rm max} + \eta_{\rm L} + s_{F,t}\quad \forall t \Big\}. \label{eq:op_flowlims}
\end{align}%
\label{eq:detailed_planning_model}%
\end{subequations}%
Coefficients $\gamma_{I,G}$, and $\gamma_{I,L}$ are the cost for upgrading generation and transmission capacity.
Cost $\gamma_{P}$ capture the necessary investments to enable EV flexibility.
Operational costs are taken as samples over operation days and scaled up to be comparable with investment cost which are computed as annual payments. 
Constraints \cref{eq:pi_lims,eq:inv_lims} are basic limits on the investments. 
The objective of the operations problem \cref{eq:op_objective} is to minimize the cost of generation plus the cost of exceeding generation and transmission capacities. Cost of generation and penalties for constraint violations are given by $c$ and $\rho_{\rm G}$, $\rho_{\rm F}$ respectively. For each time step $t$ variables $p_t$, $s_{\rm G,t}$, $s_{\rm F,t}$ collect generator production and constraint exceedance.
Constraint \cref{eq:op_enerbal} enforces the balance of production and demand at each time step. Constraint \cref{eq:op_gen_lims} softly enforces generation capacity limits for each generator which are given by the base capacities $p^{\rm max}$ and the investment decision $\eta_{\rm G}$.
Equation \cref{eq:op_dc_pf} computes power flow through a linear mapping (DC power flow via PTDF matrix $B$).
Finally, constraint \cref{eq:op_flowlims} enforces the soft constraint on the power flow limit given by the base capacities $f^{\rm max}$ and the investment decision $\eta_{\rm L}$ 


\section*{AI Disclosure}

\balance

During the preparation of this work, the authors used generative AI tools to support some coding implementation and grammar and spelling.
After any AI use, the authors reviewed and edited the content as needed and take full responsibility for the content of the published article. No AI-generated content or code was used ``as is'' and no AI has been used on the reference section.

\bibliographystyle{ieeetr}
\bibliography{ref_amps}

\endgroup
\end{document}